\documentclass[a4paper,11pt]{article}
\usepackage{amsmath}
\usepackage{amsfonts}
\usepackage{amssymb, amsthm}
\numberwithin{equation}{section}

\begin{document}

\title{Analyticity, Gevrey regularity and unique continuation for an integrable
multi-component 
peakon system with an arbitrary polynomial function}

\author{Qiaoyi Hu$^{1,2}$\footnote{e-mail:
huqiaoyi@scau.edu.cn },
~ Zhijun Qiao$^2$\footnote{e-mail: zhijun.qiao@utrgv.edu},
\\$^1$Department of Mathematics, \\South China
Agricultural University,\\Guangzhou, Guangdong 510642, PR China\\
$^2$School of Mathematical and Statistical Science,\\ University of Texas -- Rio Grande Valley,\\ Edinburg, TX 78539, USA\\
}

\date{}
\maketitle

\begin{abstract}
In this paper, we study the Cauchy problem for an integrable 
multi-component ($2N$-component) 
peakon system which is involved in an  arbitrary polynomial function.
 Based on a generalized Ovsyannikov type theorem,
 we first prove the existence and
uniqueness of solutions for the system in the Gevrey-Sobolev spaces with
 the lower bound of the lifespan. Then we show the continuity of the data-to-solution map for the system.
Furthermore, by introducing a
family of continuous diffeomorphisms of a line and utilizing the
fine structure of the system, we demonstrate  
the system exhibits unique continuation.

\end{abstract}



\section{Introduction}
\newtheorem{theorem2}{Theorem}[section]
\newtheorem{lemma2}{Lemma}[section]
\newtheorem {remark2}{Remark}[section]
\newtheorem{corollary2}{Corollary}[section]
In 2013, Xia and Qiao proposed the following integrable $2N$-component  Camassa-Holm system ($2N$CH) \cite{XQ1}
\begin{equation}\left\{\begin{array}{ll}\label{mt}
m_{j,t}=(m_jH)_{x}+m_jH\\ ~~~~~~~~~+\frac{1}{(N+1)^2}
\sum_{i=1}^{N}[m_i(u_j-u_{jx})(v_i+v_{ix})+m_j(u_i-u_{ix})(v_i+v_{ix})],\\
n_{j,t}=(n_jH)_{x}-n_jH\\ ~~~~~~~~~-\frac{1}{(N+1)^2}
\sum_{i=1}^{N}[n_i(u_i-u_{ix})(v_j+v_{jx})+n_j(u_i-u_{ix})(v_i+v_{ix})],\\
m_{j}=u_{j}-u_{jxx},\\
n_{j}=v_{j}-v_{jxx},~~1\leq j\leq N.\end{array}\right.
\end{equation}
where $H$ is an arbitrary smooth function of $u_j , v_j , 1\leq j\leq N,$ and their derivatives. In this paper, we pick  $H$ as a polynomial 
function in variables of  $u_j , v_j , 1\leq j\leq N,$ and their derivatives.
In \cite{XQ1}, the authors also showed this system 
admits a Lax pair and infinitely many conservation laws, and in particular,
its bi-Hamiltonian structures and peaked soliton (peakon) solutions were derived.
Very recently, Zhang and Yin \cite{ZY}
proved that the $2N$CH system is locally
 well-posed in Besov spaces $B^s_{p,r}$ with $1\leq p,r\leq \infty$ and $s>\max\{1-\frac{1}{p},\frac{1}{p}\}$.

The $2N$CH system is very interesting due to an arbitrary function involved so that many integrable peakon systems could be generated from its reductions. As discussed in \cite{XQ1}, in comparison with the classical soliton equations, the peakon
equations with arbitrary functions are very rare. Within our knowledge,  the $2N$CH peakon system (\ref{mt}) could be the first peakon equation with an arbitrary function.
Let us list some of its reductions below.
For instance, letting $N = 1$, $v_1 = 2$ and $H = -u_1$,
leads to the following classical Camassa-Holm equation (CH)
\begin{align}\label{CH}
m_t + 2mu_x + m_xu = 0, m = u - u_{xx},
\end{align}
which models the unidirectional propagation of shallow water waves
over a flat bottom, where $u(t,x)$ stands for the fluid velocity at
time $t$ in the spatial $x$ direction \cite{C-H,C-L,J}. The CH equation is completely integrable in the sense of Lax pair and has a bi-Hamiltonian structure 
and infinite number of conservation laws \cite{C-H, F-F, C1}. Later, the CH equation was developed to study some analysis properties and solutions on a circle \cite{C-Mc}, and extended to a whole integrable hierarchy, called the CH hierarchy, ranging from negative order to positive order \cite{Qiao-CMP} as well as having algebraic-geometric solutions \cite{FH} on a symplectic submanifold through establishing a Neumann type constraint between eigenfunctions and the CH potential function \cite{Qiao-CMP}.
In the analysis study of the CH equation,  
the Cauchy problem is locally
well-posed \cite{C4,C-E2,L-O, Rb} for the initial data $u_{0}\in
H^{s}$ with $s>\frac{3}{2}$. Moreover,
it has not only global strong solutions modelling permanent waves
\cite{C-E2,CE7,CE5}, but also blow-up solutions modelling
wave breaking \cite{C2,C-E1, C-E2, C-E3,L-O,Rb}. In addition, another amazing property of the CH equation is that there exists globally weak solutions with initial data $u_0\in
H^1$, cf. \cite{B-C,C-M,X-Z}.

If choosing $N = 1$, $u_1 = v_1$ and $H = -\frac{1}
{2} (u^2_1
- u^2 _{1,x})$,  one may obtain the following cubic nonlinear Camassa-Holm equation or modified CH ($m$CH) equation
\begin{align}\label{mCH}
m_t + \frac{1}{2}m[(u^2+u_x^2)_x ]= 0, m = u - u_{xx},
\end{align}
which is also called the FORQ equation in the literature since it was developed independently in \cite{F-F,Olver,Qiao1,Qiao2}. It might be derived from the two dimensional
Euler equations, and its Lax pair and some particular weak solutions like cuspon and peaked solutions
were addressed in \cite{Qiao1,Qiao2}. 

Selecting $N = 1$, $H=-\frac{1}{2}(u_1v_1-u_{1x}v_{1x}), u=u_1,$ and $v=v_1$ sends  Eq. (\ref{mt}) to the following integrable two-component Camassa-Holm
system with both quadratic and cubic nonlinearity proposed by Xia and Qiao \cite{Xia1}
 \begin{equation}\label{2CH1}
\left\{\begin{array}{ll}m_{t}+\frac{1}{2}[m(uv-u_xv_x)]_x-\frac{1}{2}m(uv_x-u_xv)+bu_x=0, \\
n_{t}+\frac{1}{2}[n(uv-u_xv_x)]_x+\frac{1}{2}n(uv_x-u_xv)+bv_x=0,\\
m=u-u_{xx},~n=v-v_{xx}.
\end{array}\right.
\end{equation}
Geometrically, system (\ref{2CH1}) describes a nontrivial one-parameter family of
pseudo-spherical surfaces.
 Integrability,  bi-Hamiltonian structure, and infinitely many conservation laws of the system were already developed by Xia and Qiao \cite{Xia1}. In the case $b=0$ (dispersionless
case), the authors showed that this system admits the single-peakon of travelling wave solution as well
as multi-peakon solutions.  Recently, Yan, Qiao
and Yin \cite{Yan} studied the local well-posedness in nonhomogeneous Besov spaces
 $B^s_{p,r}$ ( $1\leq p,r\leq \infty$, $s>\max\{1-\frac{1}{p},\frac{1}{p},\frac{1}{2}\}$
  but $s\neq 1+\frac{1}{p}$) for the Cauchy problem of (\ref{2CH1}) and obtained a precise
blow-up scenario and a blow-up result for the strong solutions. Moreover, Zhang and Yin \cite{ZY} presented a further improved blow-up result
for the Cauchy problem of Eq. (\ref{2CH1}) with $b=0$ provided that the initial data $m_0,n_0\in H^s$ with $s>
\frac{1}{2}$.
Very recently,
the persistence properties and unique continuation for the dispersionless ($b=0$ )
two-component system (\ref{2CH1})
were investigated by Hu and Qiao \cite{hu}.

With $v = k_1u +k_2,$ Eq. (\ref{2CH1}) is able to be reduced to the generalized CH (gCH) equation \cite{Q-X}.
The gCH equation was first implied in the work of Fokas \cite{Fokas2}. Its Lax pair, bi-Hamiltonian structure, peakons, weak kinks, kink-peakon interaction, and
classical soliton solutions were investigated in \cite{X-Q-L}.

As $N = 1$, $H = -\frac{1}{2} (u_1 - u_{1,x})(v_1 + v_{1,x}),$ $u=u_1$, and $v=v_1$,
Eq. (1.1) reads as the following two-component cubic system proposed by Song, Qu and Qiao (SQQ) \cite{SQQ}:
\begin{equation}\label{2CH2}
\left\{\begin{array}{ll}m_{t}+\frac{1}{2}[m(u-u_x)(v+v_x)]_x=0 \\
n_{t}+\frac{1}{2}[n(u-u_x)(v+v_x)]_x=0\\
m=u-u_{xx},~n=v-v_{xx}.
\end{array}\right.
\end{equation}
The SQQ equation was shown integrable in the sense of both Lax-pair and
geometry, which describes a pseudo-spherical surface. And some explicit solutions like cuspons and W/M-shape peaks solitons were discussed in \cite{SQQ} as well.
Recently, the local well-posedness
and the blow-up phenomena
for the Cauchy problem of the SQQ system were studied by Yan, Qiao and Zhang \cite{YQZ}.
Thereafter, Zhang and Yin \cite{ZY}
provided the global existence and
several new blow-up results
for the Cauchy problem of Eq. (\ref{2CH2}) through some new conservation laws they found.

The well-posedness results mentioned above are usually given in Sobolev spaces $H^s$ or Besov spaces $B^s_{p,r}$. However, the studies of analyticity and Gevrey regularity for nonlinear wave equations have already attracted lot of attention and interest. 
There have been plenty of literature concerning these issues.
For instance, the hydrodynamics of Euler equations was initiated by Ovsyannikov \cite{O1,O2} and later developed with a further study in \cite{CE2,Nr,Ns,Tre1} and in \cite{B-G1,B-G2} where the approach is based on a contraction type argument in a suitable scale of Banach spaces.
Constantin and Escher \cite{C-E6} studied the analyticity of periodic traveling free
surface water waves with vorticity.
For the classical
Camassa-Holm equation on the circle, a result similar to Theorem 3.1 in our paper was proved in \cite{H-M}, but without an analytic
lifespan estimate like (\ref{TT}).  The analyticity of the Cauchy problem for two-component Camassa-Holm
shallow water
was studied in \cite{YY1}.
Recently, Barostichi, Himonas and Petronilho \cite{B-H-P1} established the well-posedness for a class of nonlocal evolution equations 
in spaces of analytic functions. Furthermore, they proved a Cauchy-Kovalevsky theorem for a so-called generalized Camassa-Holm equation in \cite{B-H-P}.
Very recently, Luo and Yin \cite{Luo} investigated the Gevrey regularity and analyticity for a class of Camassa-Holm type systems including a three-component Camassa-Holm system
 derived by Geng and Xue
 \cite{GX}, a two-component shallow water system
 proposed by Constantin and Ivanov
  \cite{C-I}, and a modified two-component Camassa-Holm system found by Holm, Naraigh and Tronci \cite{HNT}.

To the best of our knowledge, the analyticity and the Gevrey regularity for the $2N$CH system (\ref{mt}) with an arbitrary
function has not been studied yet. In this paper, our major goal aims to construct
 the well-posedness of the Cauchy problem for the $2N$CH system (\ref{mt}) in a suitable
 Gevrey-Sobolev spaces. 
Our approach is mainly motivated from the Cauchy-Kovalevsky type results in \cite{B-G2,B-H-P1,Luo}.
As discussed above, different
 choices of $H$ in the $2N$CH peakon system (\ref{mt}) yield various amazing CH type equations, thus
Eq. (\ref{mt})  deserves a deep investigation.
 Employing the good symmetric structure of the $2N$CH system (\ref{mt}) and a
family of continuous diffeomorphisms to the line,
 we may also 
study the unique continuation of the system with $N=1$.


\par
The whole paper is organized as follows.
In Section 2, 
we briefly give 
preliminary results including the abstract Ovsyannikov type theorem and the
basic properties of the Gevrey-Sobolev space $G^{\delta,s,b}$ (see Eqs. (2.2) and (2.3)).
In Section 3,  
we prove the existence and
uniqueness of Gevrey regularity of the solution to the
$2N$CH system with an estimate about the lifespan.
In Section 4, we study the continuity of the data-to-solution
map in spaces of Gevrey class functions, and present the local well-posedness
in Gevrey-Sobolev spaces for two
integrable peakon systems with arbitrary
functions. In the last section, we show 
that the $2N$CH system exhibits the unique continuation if $N=1$. Particularly, the unique continuation
for the SQQ system ({\ref{2CH2}) is illustrated.

\section{Preliminaries}
\newtheorem{theorem3}{Theorem}[section]
\newtheorem{lemma3}{Lemma}[section]
\newtheorem {remark3}{Remark}[section]
\newtheorem{corollary3}{Corollary}[section]
\newtheorem{definition3}{Definition}[section]
In this section,
we first present an abstract generalized Ovsyannikov type theorem in a scale of decreasing Banach spaces. Through recalling various analytic and Gevrey class norms, we then define the Gevrey-Sobolev spaces $G^{\delta,s,b}$. Furthermore, some useful properties of $G^{\delta,s,b}$ are discussed.

Let us first recall the definition of a scale of Banach spaces $\{X_{\delta}\}_{0<\delta\leq 1}$.
\begin{definition3}
A scale of Banach spaces is defined as a one-parameter family of
Banach spaces $\{X_{\delta}\}_{0<\delta\leq 1}$ satisfying

(Scale): For any $0 <\delta'< \delta \leq  1$
\begin{equation}
X_{\delta}\subset X_{\delta'},~~\|\cdot\|_{\delta'}\leq \|\cdot\|_{\delta}.
\end{equation}
\end{definition3}

We then provide a brief description of a generalized autonomous
Ovsyannikov theorem \cite{Luo}, which came up from \cite{O1,O2}.

Given a decreasing scale of Banach spaces $\{X_{\delta}\}_{0<\delta\leq 1}$ and initial data $w_0\in  X_1$, let us consider the Cauchy problem
\begin{equation}\label{Fu}
\frac{dw}{dt}=F(w),~~w(0)=w_0
\end{equation}
where $F:~X_{\delta}\rightarrow X_{\delta'}$ satisfies the following conditions.

(1). Let $F:~X_{\delta}\rightarrow X_{\delta'}$ be a function. Then for any given $w_0\in  X_1$ and $R>0$ there exist positive constants $L$ and $M$, depending on $w_0$ and $R$, such that for all
$0 < \delta' < \delta < 1$ and $w_1,w_2 \in X_{\delta}$ with $\|w_1-w_0\|_{\delta}<R$ and $\|w_2-w_0\|_{\delta}<R$, $F$ satisfies 
the following Lipschitz type condition
\begin{equation}\label{con1}
\|F(w_1)-F(w_2)\|_{\delta'}\leq \frac{L}{(\delta-\delta')^b}\|w_1-w_2\|_{\delta},
\end{equation}
with the following bound for the $X_{\delta}$ norm of $F(w_0)$
\begin{equation}\label{con2}
\|F(w_0)\|_{\delta}\leq \frac{M}{(1-\delta)^b},~0<{\delta}<1.
\end{equation}

(2). Let $0<\delta'<{\delta}<1$ and $T>0$. If the function $t\mapsto w(t)$ is holomorphic on $\{t \in \mathbb{C} :
|t| < T\}$ with values in $X_{\delta}$ and $\sup_{|t| <T} \|w(t) - w_0\|_{\delta} < R$, then the function
$t \mapsto F(w(t))$ is holomorphic on $\{t \in \mathbb{C} :
|t| < T\}$ with values in $X_{\delta'}.$

Let us now present the generalized Ovsyannikov theorem.
\begin{theorem3}\cite{Luo}\label{abstract}
Let $R > 0$, $b\geq1$ and $T> 0$. Assume that the scale of Banach spaces
$X_{\delta}$ and the function $F(w)$ satisfy the above conditions (1) and (2). Then, for given $w_0 \in X_1$, there exists a positive number $T'$ such that $0<T'<T$ with the lower bound of the lifespan given by
\begin{equation}\label{TT}
T' =\min\{\frac{(2^{b}-1)R}{(2^{b}-1)2^{2b+3}LR + M},\frac{1}{2^{2b+4}L}\},
\end{equation}
and a unique solution $u(t)$ to the Cauchy problem (\ref{Fu}), which is a holomorphic
function for every $ \delta \in (0, 1)$ in the disc $D(0, \frac{T'(1 - \delta)^b}{2^b-1})$ valued in $X_{\delta}$.
\end{theorem3}

\begin{remark3}
As illustrated in \cite{Luo}, the generalized Ovsyannikov theorem 2.1 covers the classical Ovsyannikov theorem as its special case (with $b=1$) \cite{Tre1,Tre2,Tre3}. Various versions of the Cauchy-Kovalevsky theorem were established by Baouendi and Goulaouic \cite{B-G1,B-G2}, Nirenberg \cite{Nr}, and
Nishida \cite{Ns}.
\end{remark3}

Next we recall some analytic and Gevrey-class norms. The classical definition is given as follow.
\begin{definition3}\label{def1} 
A $C^{\infty}(\mathbb{T}^d)$ function $f$ is in the Gevrey-class $b$,
for some $b > 0$ if there exist $M, \delta > 0$ such that
$$|\partial^{s}f|\leq M \frac{s!^b}{\delta^{|s|}},$$
for all $x \in \mathbb{T}^d$, and all multi-indices $s$.
$\delta$ is the radius of Gevrey-class
regularity of the function $f$.
\end{definition3}
\begin{remark3}
If $0 < b < 1$, then $f$ is called an ultra-analytic function. If
$b = 1$, $f$ recovers the usual analytic function and $\delta$ is (up to a dimensional constant) the radius of convergence of the Taylor
series at each point.
 If $b > 1$, then $f$ is called the Gevrey-class $b$ function, which consists of $C^{\infty}$ function but is not
analytic.
\end{remark3}

It is however more convenient in PDEs to use an equivalent characterization, introduced
by Foias and Temam \cite{FT} to address the analyticity of solutions of the Navier-Stokes equations.
\begin{definition3} \cite{FT}\label{def2}
For all $b \geq 1$ the Gevrey-class $b$ is given by
\begin{align*}
\bigcup\limits_{\delta  > 0} {D(\Lambda ^s e^{\delta \Lambda ^{1/b} } )}
\end{align*}
for any $s \geq 0$, where
\begin{align}
\|\Lambda ^s e^{\delta \Lambda ^{1/b} } f\|_{L^2}^2=(2\pi)^3\sum\limits_{k \in Z^d }^{} {|k|^{2s} e^{2\delta |k|^{1/b}  }|\hat{f}(k)|^2 },
\end{align}
where $\widehat{f}$ is the Fourier coefficient of $f$ in $\mathbb{T}^d$.
\end{definition3}

In order to obtain the local well-posedness result for the Cauchy problem of the $2N$CH system (1.1) by using Theorem \ref{abstract}, we are now introduce the framework of the Gevery class spaces $G^{\delta,s,b}$ by combining the above two definitions \ref{def1} and \ref{def2}.
In the following contexts, we denote by $\widehat{f}$ the Fourier transform of $f$ in $\mathbb{R}$ or $\mathbb{T}$.
The Gevrey class spaces $G^{\delta,s,b}$ are defined below:
\begin{equation}\label{GT}
G^{\delta,s,b}(\mathbb{T})=\{f\in C^{\infty}(\mathbb{T}): \|f\|_{G^{\delta,s,b}(\mathbb{T})}^2:=\sum_{k\in \mathbb{Z}}(1+|k|^{2})^{s}e^{2\delta|k|^{\frac{1}{b}}}|\hat{f}(k)|^2<\infty\},
\end{equation}
where $\delta > 0$, $s \geq 0$, and $b>0$ for the periodic case.
While in the non-periodic case, the 
Banach spaces $G^{\delta,s,b}$ are defined by
\begin{equation}\label{GR}
G^{\delta,s,b}(\mathbb{R})=\{f\in C^{\infty}(\mathbb{R}): \|f\|_{G^{\delta,s,b}(\mathbb{R})}^2:=\int_{\mathbb{R}}(1+|\xi|^{2})^{s}e^{2\delta|\xi|^{\frac{1}{b}}}|\hat{f}(\xi)|^2d\xi<\infty\}.
\end{equation}

\begin{remark3}
If $f\in  G^{\delta,s,b}$, then as per the above definitions, the following properties hold. 

(i) If $0 <\delta'< \delta $, $s \geq 0$ and $b > 0$, then $\|\cdot\|_{\delta',s,b}\leq \|\cdot\|_{\delta,s,b},$ i.e. $G^{\delta,s,b}\hookrightarrow G^{\delta',s,b}$.

(ii) If $0 \leq s'< s $, $\delta > 0$ and $b > 0$, then $\|\cdot\|_{\delta,s',b}\leq \|\cdot\|_{\delta,s,b},$ i.e. $G^{\delta,s,b}\hookrightarrow G^{\delta,s',b}$.

(iii) If $0 <b'< b $, $\delta > 0$ and $s \geq 0$, then $\|\cdot\|_{\delta,s,b}\leq \|\cdot\|_{\delta,s,b'},$ i.e. $G^{\delta,s,b'}\hookrightarrow G^{\delta,s,b}$.
\end{remark3}

\begin{remark3}
By Remark 2.3 (i), one may easily see that 
$$\{G^{\delta,s,b}\}_{0< \delta <1}, ~~with~ norm \| \cdot \|_{\delta,s,b}$$
form a scale of decreasing Banach spaces.
\end{remark3}

In the following contexts, if no specified otherwise,
$\|\cdot\|_{\delta,s,b}$ stands for the norm, and $G^{\delta,s,b}$ for the space in both periodic and non-periodic cases when a result holds for both cases.






Finally, let us summarize some basic properties
of the $G^{\delta,s,b}$ spaces, which were mentioned in \cite{Luo}.


\begin{lemma3}
If $0 <\delta'< \delta \leq  1$, $s \geq 0$, $b>0$ and $f\in  G^{\delta,s,b}$, then
\begin{align}
\label{2.4}&\|f_x\|_{\delta',s,b}\leq \frac{e^{-b}b^b}{(\delta-\delta')^b}\|f\|_{\delta,s,b}\\
&\|f_x\|_{\delta,s,b}\leq \|f\|_{\delta,s+1,b}\\
\label{2.6}&\|(1-\partial_x^2)^{-1}f\|_{\delta,s+2,b}=\|f\|_{\delta,s,b}\\
&\|(1-\partial_x^2)^{-1}f\|_{\delta,s,b}\leq \|f\|_{\delta,s,b}\\
&\|\partial_x(1-\partial_x^2)^{-1}f\|_{\delta,s,b}\leq \|f\|_{\delta,s,b}
\end{align}
\end{lemma3}

\begin{lemma3} (Algebraic property)

If $0 <\delta< 1 $, $s>\frac{1}{2}$, $b>0$ and $f,g\in  G^{\delta,s,b}$, then we have
\begin{equation}\label{algebra}
\|fg\|_{\delta,s,b}\leq C_s \|f\|_{\delta,s,b}\|g\|_{\delta,s,b}.
\end{equation}
\end{lemma3}




\section{Existence and uniqueness}
\newtheorem{theorem4}{Theorem}[section]
\newtheorem{lemma4}{Lemma}[section]
\newtheorem {remark4}{Remark}[section]
\newtheorem{corollary4}{Corollary}[section]
\par
In this section, we consider the Cauchy problem for the multi-component integrable Camassa-Holm system (1.1).
In the following contexts, constant $C$ refers to a generic positive constant,
which may change from one to another.

Let $U=(u_1,...,u_N)^T$, $V=(v_1,...,v_N)^T$, $U_t=(u_{1,t},...,u_{N,t})^T$,
 $V_t=(v_{1,t},...,v_{N,t})^T$, $U_x=(u_{1,x},...,u_{N,x})^T$,
and $V_x=(v_{1,x},...,v_{N,x})^T$ denote the vector functions and their partial derivatives with respective to the time $t$ and spacial variable $x$. In our current work, we focus on the Hamiltonian $H=H(U,V,U_x,V_x)$ in the system (1.1) being a polynomial function with degree $k$, which of course does not affect integrability of (1.1).
Let
\begin{align}\label{fj}
\nonumber F_j&:=(m_jH)_{x}+m_jH\\&~~~~~~+\frac{1}{(N+1)^2}
\sum_{i=1}^{N}[m_i(u_j-u_{jx})(v_i+v_{ix})+m_j(u_i-u_{ix})(v_i+v_{ix})]
\end{align}
 and
\begin{align}\label{gj}
\nonumber G_j&:=(n_jH)_{x}-n_jH\\&~~~~~~-\frac{1}{(N+1)^2}
\sum_{i=1}^{N}[n_i(u_i-u_{ix})(v_j+v_{jx})+n_j(u_i-u_{ix})(v_i+v_{ix})].
\end{align}
As per $m_j=(1-\partial_x^2)u_j$ and $n_{j}=(1-\partial_x^2)v_{j}$,  a direct
calculation makes one propose the following Cauchy problem for Eq.(1.1):
\begin{equation}\left\{\begin{array}{ll}\label{mtivp}
u_{j,t}=(1-\partial_x^2)^{-1}F_j,\\
n_{j,t}=(1-\partial_x^2)^{-1}G_j,\\
u_{j}(0,x)=u_{j,0},\\
v_{j}(0,x)=v_{j,0},~~1\leq j\leq N.\end{array}\right.
\end{equation}

To apply the generalized abstract Ovsyannikov Theorem 2.1 conveniently,  let us rewrite the system in the following matrix form:
$W=\left( \begin{array}{l}
 U \\
 V \\
 \end{array} \right)$, $W_x=\left( \begin{array}{l}
 U_x \\
 V_x \\
 \end{array} \right)$, $U_0=(u_{1,0},...,u_{n,0})^T$, $V_0=(v_{1,0},...,v_{n,0})^T$, $W_0=\left( \begin{array}{l}
 U_0 \\
 V_0 \\
 \end{array} \right)$, $F=(F_1,...F_n)^T,~~G=(G_1,...G_n)^T$, $P=\left( \begin{array}{l}
 F \\
 G \\
 \end{array} \right)$ and , then (\ref{mtivp}) reads as
\begin{equation}\left\{\begin{array}{ll}\label{mtivp1}
W_t=(1-\partial_x^2)^{-1}P(W,W_x),\\
W(0,x)=W_{0}.
\end{array}\right.
\end{equation}

For any $2N-$component vector function $W=(u_1,u_2,...,u_N,v_1,v_2,...,v_N)^T\in (G^{\delta,{s},b})^{2N}$, we define its norm as follows:
\begin{align}\label{W}
\|W\|_{\delta,s,b}=\sum_{j=1}^{N}(\|u_{j}\|_{\delta,s,b}+\|v_{j}\|_{\delta,s,b}).
\end{align}
For our convenience, 
let us assume that
 the initial data $W_0\in (G^{1,{s+2},b})^{2N}$, then we have the following theorem for the existence and uniqueness of Eq. (\ref{mtivp1}).

\begin{theorem3}\label{local}
Let $s > \frac{1}{2},~~b\geq 1.$ If $W_0 \in (G^{1,{s+2},b})^{2N}$,
then there exists a positive
time $T'$, which depends on the initial data $W_0$, such that for every $\delta \in (0, 1),$ the Cauchy
problem (\ref{mtivp1}) has a unique solution $W(t)=(u_1(t),...,u_N(t),v_1(t),...,v_N(t))^T$, which is a holomorphic vector function in the disc $D(0, \frac{T'(1-\delta)^b}{2^b-1})$
valued
in $(G^{\delta,{s},b})^{2N}$,
 where 
 $T'\approx \frac{1}{\|W_0\|_{\delta,s+2,b}^k+\|W_0\|_{\delta,s+2,b}^2}$.
\end{theorem3}

For given $s\geq 0$ and $b> 0$, Remark 2.3 ensures $G^{\delta,s,b}$ satisfying the scale decreasing condition (2.1) like the spaces $X_{\delta}$ in the generalized
Ovsyannikov theorem. Also, these spaces and $(1-\partial_x^2)^{-1}P(W,W_x)$ satisfy the condition (2).
Therefore, to prove
Theorem 3.1 it suffices to show that the right-hand side of Eq. (\ref{mtivp1}) satisfies conditions (\ref{con1}) and (\ref{con2}), which are shown in the following crucial Lemma.

\begin{lemma4} Let $s > \frac{1}{2}$,  $b\geq 1$, $R > 0$,   and $W_0 \in (G^{1,s+2,b})^{2N}$.
 Then, for the Cauchy problem (\ref{mtivp1}) there exist two positive constants $L$ and $M$,
 which depend on $R$ and $\|W_0\|_{1,s+2,b}$,
such that for $W_1,W_2 \in (G^{1,s+2,b})^{2N}$, $\|W_{1}-W_0\|_{\delta,s+2,b} < R,~ \|W_2-W_0\|_{\delta,s+2,b} < R $
and $0 < \delta'<\delta <1$ we have
\begin{equation}\label{P1P2}
\|(1-\partial_x^2)^{-1}[P(W_1,W_{1,x}) - P(W_2,W_{2,x})]\|_{\delta',s+2,b}\leq \frac{L}{(\delta-\delta')^b} \|W_1-W_2\|_{\delta,s+2,b}
\end{equation}
and
\begin{equation}\label{P0}
\|(1-\partial_x^2)^{-1}P(W_0,W_{0,x})\|_{\delta,s+2,b}\leq \frac{M}{(1-\delta)^b}.
\end{equation}
Furthermore, the lower bound of the lifespan $T'$ can be estimated through
\begin{equation}
T'\approx \frac{1}{\|W_0\|_{\delta,s+2,b}^k+\|W_0\|_{\delta,s+2,b}^2}.
\end{equation}
\end{lemma4}

\begin{proof}
To prove (\ref{P1P2}) and (\ref{P0}), we just need to prove for $1\leq j \leq N$,
$(1-\partial_x^2)^{-1}F_j$ and $(1-\partial_x^2)^{-1}G_j$ satisfy (\ref{con1}) and (\ref{con2}).

Let $W_1=(u_1^{(1)},...,u_N^{(1)},v_1^{(1)},...,v_N^{(1)})^T$ and
$W_2=(u_1^{(2)},...,u_N^{(2)},v_1^{(2)},...,v_N^{(2)})^T$. For $j=1,...,N, k=1,2$, we denote
 \begin{align*}
~~~~~m_j^{(k)}=(1-\partial_x^2)u_j^{(k)},
\end{align*}
$$
f_j^{(k)}:=\frac{1}{(N+1)^2}
\sum_{i=1}^{N}[m_i^{(k)}(u_j^{(k)}-u^{(k)}_{jx})(v^{(k)}_i+v^{(k)}_{ix})+m^{(k)}_j
(u^{(k)}_i-u^{(k)}_{ix})(v^{(k)}_i+v^{(k)}_{ix})],
$$
and $H^{(k)}$ are the polynomial functions of $W_k$ and its derivatives $W_{k,x}$.
Applying the triangle inequality leads to 
\begin{align*}
~~~~~\|(1-\partial_x^2)^{-1}[F_j(W_1,W_{1,x})-F_j(W_2,W_{2,x})\|_{\delta',s+2,b}
\leq I+II+III,
\end{align*}
where
\begin{align*}
~~~~~I&=\|(1-\partial_x^2)^{-1}(m_j^{(1)}H^{(1)}-m_j^{(2)}H^{(2)})_x\|_{\delta',s+2,b}\\
~~~~~II&=\|(1-\partial_x^2)^{-1}(m_j^{(1)}H^{(1)}-m_j^{(2)}H^{(2)})\|_{\delta',s+2,b}\\
~~~~~III&=\|(1-\partial_x^2)^{-1}(f_j^{(1)}-f_j^{(2)})\|_{\delta',s+2,b}\\
\end{align*}
To estimate the above three terms, let us establish the following inequalities.
For any $1\leq j\leq N$, we have
\begin{align}\label{mj}
\nonumber\|m_j^{(i)}\|_{\delta,s,b}&=\|(1-\partial_x^2)u_j^{(i)}\|_{\delta,s,b}\\
&=\|u_j^{(i)}\|_{\delta,s+2,b}\leq \|W_i\|_{\delta,s+2,b},i=1,2.
\end{align}
\begin{align}\label{mj1mj2}
\nonumber\|m_j^{(1)}-m_j^{(2)}\|_{\delta,s,b}&=\|(1-\partial_x^2)(u_j^{(1)}-u_j^{(2)})\|_{\delta,s,b}\\
&=\|u_j^{(1)}-u_j^{(2)}\|_{\delta,s+2,b}\leq \|W_1-W_2\|_{\delta,s+2,b}.
\end{align}
Since $H=H(W,W_x)$ is a polynomial function with degree $k$, we may obtain 
\begin{align}\label{H1H2}
\nonumber\|H^{(1)}-H^{(2)}\|_{\delta,s,b}&\leq C(\|W_1-W_2\|_{\delta,s,b}+\|W_{1,x}-W_{2,x}\|_{\delta,s,b})\\
\nonumber&~~(\|W_1\|^{k-1}_{\delta,s,b}+\|W_2\|^{k-1}_{\delta,s,b}+\|W_{1,x}\|^{k-1}_{\delta,s,b}
+\|W_{2,x}\|^{k-1}_{\delta,s,b})\\
\nonumber &\leq C(\|W_1-W_2\|_{\delta,s,b}+\|W_{1}-W_{2}\|_{\delta,s+1,b})\\
\nonumber&~~(\|W_1\|^{k-1}_{\delta,s,b}+\|W_2\|^{k-1}_{\delta,s,b}+\|W_{1}\|^{k-1}_{\delta,s+1,b}
+\|W_{2}\|^{k-1}_{\delta,s+1,b})\\
&\leq C\|W_1-W_2\|_{\delta,s+2,b}(\|W_{1}\|^{k-1}_{\delta,s+2,b}+\|W_{2}\|^{k-1}_{\delta,s+2,b}).
\end{align}
Casting $H^{(1)}=0$ or $H^{(2)}=0$ in (\ref{H1H2}) gives
 \begin{align}
 \|H^{(i)}\|_{\delta,s,b}\leq C \|W_i\|^k_{\delta,s+2,b},i=1,2.
\end{align}
For $i=1,2$, we have
\begin{align}\label{uj}
\nonumber\|W_i\|_{\delta,s+2,b}&\leq \|W_i-W_{0}\|_{\delta,s+2,b}+\|W_{0}\|_{\delta,s+2,b}\\
&\leq R+\|W_{0}\|_{1,s+2,b}.
\end{align}

Let us now estimate $I,\ II$, and $III$. Eqs. (\ref{mj})-(\ref{uj}), (\ref{2.4}), (\ref{2.6}) and (\ref{algebra}) lead to
\begin{align}\label{I}
\nonumber I&=\|(1-\partial_x^2)^{-1}(m_j^{(1)}H^{(1)}-m_j^{(2)}H^{(2)})_x\|_{\delta',s+2,b}\\
\nonumber &=\|(m_j^{(1)}H^{(1)}-m_j^{(2)}H^{(2)})_x\|_{\delta',s,b}\\
\nonumber&\leq \frac{e^{-b}b^b}{(\delta-\delta')^b}\|m_j^{(1)}H^{(1)}-m_j^{(2)}H^{(2)}\|_{\delta,s,b}\\
\nonumber&\leq \frac{e^{-b}b^b}{(\delta-\delta')^b} C_s(\|m_j^{(1)}\|_{\delta,s}\|H^{(1)}-H^{(2)}\|_{\delta,s,b}+
\|H^{(2)}\|_{\delta,s,b}\|m_j^{(1)}-m_j^{(2)}\|_{\delta,s})\\
\nonumber&\leq\frac{C}{(\delta-\delta')^b} [\|W_1\|_{\delta,s+2,b}\|W_1-W_2\|_{\delta,s+2,b}
(\|W_1\|^{k-1}_{\delta,s+2,b}+\|W_2\|^{k-1}_{\delta,s+2,b})\\
\nonumber&~~+\|W_2\|^{k}_{\delta,s+2,b}\|W_1-W_2\|_{\delta,s+2,b}]\\
\nonumber&\leq\frac{C}{(\delta-\delta')^b}(\|W_1\|_{\delta,s+2,b}+\|W_2\|_{\delta,s+2,b})^{k}
\|W_1-W_2\|_{\delta,s+2,b}\\
&\leq \frac{C}{(\delta-\delta')^b}(R+\|W_0\|_{1,s+2,b})^{k}
\|W_1-W_2\|_{\delta,s+2,b}.
\end{align}

Employing Eqs. (\ref{mj})-(\ref{uj}), (\ref{2.6}) and (\ref{algebra}) yields
\begin{align}\label{II}
\nonumber II&=\|(1-\partial_x^2)^{-1}(m_j^{(1)}H^{(1)}-m_j^{(2)}H^{(2)})\|_{\delta',s+2,b}\\
\nonumber &=\|m_j^{(1)}H^{(1)}-m_j^{(2)}H^{(2)}\|_{\delta',s,b}\\
\nonumber&\leq C_s(\|m_j^{(1)}\|_{\delta,s}\|H^{(1)}-H^{(2)}\|_{\delta,s,b}+
\|H^{(2)}\|_{\delta,s,b}\|m_j^{(1)}-m_j^{(2)}\|_{\delta,s})\\
\nonumber&\leq C [\|W_1\|_{\delta,s+2,b}\|W_1-W_2\|_{\delta,s+2,b}
(\|W_1\|^{k-1}_{\delta,s+2,b}+\|W_2\|^{k-1}_{\delta,s+2,b})\\
\nonumber&~~+\|W_2\|^{k}_{\delta,s+2,b}\|W_1-W_2\|_{\delta,s+2,b}]\\
\nonumber&\leq C (\|W_1\|_{\delta,s+2,b}+\|W_2\|_{\delta,s+2,b})^{k}
\|W_1-W_2\|_{\delta,s+2,b}\\
\nonumber&\leq C(R+\|W_0\|_{\delta,s+2,b})^{k}
\|W_1-W_2\|_{\delta,s+2,b}\\
&\leq \frac{C}{(\delta-\delta')^b}(R+\|W_0\|_{1,s+2,b})^{k}
\|W_1-W_2\|_{\delta,s+2,b},
\end{align}
where we use $0<\delta'<\delta<1$ and $b\geq1$, which implies $0<(\delta-\delta')^b<1$.

Let
\begin{align}\label{fj}
\nonumber f_j(W,W_x)&:=\frac{1}{(N+1)^2}
\sum_{i=1}^{N}[m_i(u_j-u_{jx})(v_i+v_{ix})+m_j(u_i-u_{ix})(v_i+v_{ix})]\\
&=l_j+z_j, \ 1\leq j \leq N,
\end{align}
where 
\begin{align}\label{lj}
l_j(W,W_x):=\frac{1}{(N+1)^2}
\sum_{i=1}^{N}[m_i(u_j-u_{jx})(v_i+v_{ix})],
\end{align}
and
\begin{align}\label{zj}
z_j(W,W_x):=\frac{1}{(N+1)^2}
\sum_{i=1}^{N}[m_j(u_i-u_{ix})(v_i+v_{ix})].
\end{align}
To estimate $III$, let us first give some essential inequalities below:
\begin{align}\label{l1l2}
\nonumber&~~~~\|l_j^{(1)}-l_j^{(2)}\|_{\delta',s,b}\\
\nonumber&=\|\frac{1}{(N+1)^2}
\sum_{i=1}^{N}[m_i^{(1)}(u_j^{(1)}-u_{jx}^{(1)})(v_i^{(1)}+v_{ix}^{(1)})
-m_i^{(2)}(u_j^{(2)}-u_{jx}^{(2)})(v_i^{(2)}+v_{ix}^{(2)})]\|_{\delta,s,b}\\
\nonumber&\leq C\sum_{i=1}^{N}[\|m_i^{(1)}\|_{\delta,s,b}\|u_j^{(1)}-u_{jx}^{(1)}\|_{\delta,s,b}
(\|v_i^{(1)}-v_{i}^{(2)}\|_{\delta,s,b}+\|v_{i,x}^{(1)}-v_{ix}^{(2)}\|_{\delta,s,b})\\
\nonumber&~~~~~~~~~~~+\|m_i^{(1)}\|_{\delta,s,b}\|v_i^{(2)}-v_{ix}^{(2)}\|_{\delta,s,b}
(\|u_j^{(1)}-u_{j}^{(2)}\|_{\delta,s,b}+\|u_{j,x}^{(1)}-u_{jx}^{(2)}\|_{\delta,s,b})\\
\nonumber&~~~~~~~~~~~+\|u_j^{(2)}-u_{jx}^{(2)}\|_{\delta,s,b}
\|v_i^{(2)}+v_{i,x}^{(2)}\|_{\delta,s,b}\|m_{i}^{(1)}-m_{i}^{(2)}\|_{\delta,s,b}]\\
\nonumber &\leq C\sum_{i=1}^{N}[\|u_i^{(1)}\|_{\delta,s+2,b}(\|u_j^{(1)}\|_{\delta,s,b}+
\|u_{jx}^{(1)}\|_{\delta,s,b})
(\|W_1-W_{2}\|_{\delta,s,b}+\|W_{1,x}-W_{2x}\|_{\delta,s,b})\\
\nonumber&~~~~~~~~~~~+\|u_i^{(1)}\|_{\delta,s+2,b}(\|v_i^{(2)}\|_{\delta,s,b}
+\|v_{ix}^{2}\|_{\delta,s,b})
(\|W_1-W_{2}\|_{\delta,s,b}+\|W_{1,x}-W_{2x}\|_{\delta,s,b})\\
\nonumber&~~~~~~~~~~~+(\|u_j^{(2)}\|_{\delta,s,b}+\|u_{jx}^{(2)}\|_{\delta,s,b})
(\|v_i^{(2)}\|_{\delta,s,b}+\|v_{i,x}^{(2)}\|_{\delta,s,b})
\|u_{i}^{(1)}-u_{i}^{(2)}\|_{\delta,s+2,b}]\\
\nonumber&\leq C(\|W_1\|_{\delta,s+2,b}+\|W_2\|_{\delta,s+2,b})^2\|W_1-W_2\|_{\delta,s+2,b}\\
&\leq \frac{C}{(\delta-\delta')^b}(R+\|W_0\|_{1,s+2,b})^2\|W_1-W_2\|_{\delta,s+2,b}, \ 1\leq j \leq N.
\end{align}
Similarly, we have
\begin{align}\label{z1z2}
\|z_j^{(1)}-z_j^{(2)}\|_{\delta',s,b}\leq \frac{C}{(\delta-\delta')^b}(R+\|W_0\|_{1,s+2,b})^2\|W_1-W_2\|_{\delta,s+2,b},~~1\leq j\leq N.
\end{align}

Hence, (\ref{2.6}), (\ref{l1l2}) and (\ref{z1z2}) imply
\begin{align}\label{III}
\nonumber III&=\|(1-\partial_x^2)^{-1}(f_j^{(1)}-f_j^{(2)})\|_{\delta',s+2,b}\\
\nonumber&=\|f_j^{(1)}-f_j^{(2)}\|_{\delta',s,b}\\
\nonumber&\leq \|l_j^{(1)}-l_j^{(2)}\|_{\delta',s,b}+\|z_j^{(1)}-z_j^{(2)}\|_{\delta',s,b}\\
&\leq \frac{C}{(\delta-\delta')^b}(R+\|W_0\|_{1,s+2,b})^{2}
\|W_1-W_2\|_{\delta,s+2,b}.
\end{align}

Adding (\ref{I}), (\ref{II}) and (\ref{III}) reveals
\begin{align}\label{FF}
\nonumber&~~~~\|(1-\partial_x^2)^{-1}(F_j(W_1,W_{1,x})-F_j(W_2,W_{2,x}))\|_{\delta',s+2,b}\\
&\leq \frac{L}{(\delta-\delta')^b}\|W_1-W_2\|_{\delta,s+2,b},
\end{align}
where
\begin{align}\label{L}
L=C[(R+\|W_0\|_{1,s+2,b})^{k}+(R+\|W_0\|_{1,s+2,b})^{2}].
\end{align}
Analogous to the proof of (\ref{FF}), we can obtain
\begin{align}\label{GG}
\nonumber&~~~~\|(1-\partial_x^2)^{-1}(G_j(W_1,W_{1,x})-G_j(W_2,W_{2,x}))\|_{\delta',s+2,b}\\
&\leq \frac{L}{(\delta-\delta')^b}\|W_1-W_2\|_{\delta,s+2,b},
\end{align}
where $L$ is defined in (\ref{L}).
Therefore, (\ref{P1P2}) is desired.

Next, we derive (\ref{P0}). For $1\leq j \leq N$, we have
\begin{align*}
~~~~~\|(1-\partial_x^2)^{-1}[F_j(W_0,W_{0,x})]\|_{\delta',s+2,b}
\leq J_1+J_2+J_3,
\end{align*}
where
\begin{align*}
J_1&=\|(1-\partial_x^2)^{-1}(m_{j,0}H_0)_x\|_{\delta',s+2,b},\\
J_2&=\|(1-\partial_x^2)^{-1}(m_{j,0}H_0)\|_{\delta',s+2,b},\\
J_3&=\|(1-\partial_x^2)^{-1}f_{j,0}\|_{\delta',s+2,b}= \|f_{j,0}\|_{\delta',s,b}.
\end{align*}
Because $H_0$ is a degree $k$ polynomial of $W_0$ and $W_{0,x}$, we have
\begin{align}\label{H0}
\|H_0(W_0,W_{0,x})\|_{\delta,s,b}\leq C \|W_0\|^k_{1,s+2,b}.
\end{align}
Combining (\ref{2.4}), (\ref{2.6}), and (\ref{algebra}) with (\ref{H0}) generates
\begin{align}\label{J1}
\nonumber J_1&=\|(1-\partial_x^2)^{-1}(m_{j,0}H_0)_x\|_{\delta',s+2,b}\\
\nonumber &=\|(m_{j,0}H_0)_x\|_{\delta',s,b}\\
\nonumber&\leq \frac{e^{-b}b^b}{(\delta-\delta')^b}\|m_{j,0}H_0\|_{\delta,s,b}\\
\nonumber&\leq \frac{e^{-b}b^b}{(\delta-\delta')^b}C_s\|m_{j,0}\|_{\delta,s,b}\|H_0\|_{\delta,s,b}\\
\nonumber&\leq \frac{C}{(\delta-\delta')^b}\|u_{j,0}\|_{\delta,s+2,b}\|W_0\|^k_{\delta,s+2,b}\\
&\leq \frac{C}{(\delta-\delta')^b}\|W_0\|_{1,s+2,b}^{k+1}.
\end{align}
As per (\ref{2.6}) and (\ref{algebra}), we obtain
\begin{align}\label{J2}
\nonumber J_2&=\|(1-\partial_x^2)^{-1}(m_{j,0}H_0)\|_{\delta',s+2,b}\\
\nonumber&=\|m_{j,0}H_0\|_{\delta',s,b}\\
\nonumber&\leq C_s\|m_{j,0}\|_{\delta',s,b}\|H_0\|_{\delta',s,b}\\
\nonumber&= C_s\|u_{j,0}\|_{\delta',s+2,b}\|H_0\|_{\delta',s,b}\\
\nonumber&\leq C\|W_0\|_{\delta,s+2,b}\|W_0\|^k_{\delta,s+2,b}\\
&\leq \frac{C}{(\delta-\delta')^b}\|W_0\|_{1,s+2,b}^{k+1}.
\end{align}
Let us now estimate $J_3$. To do so, we need  the following inequalities:
\begin{align}\label{lj0}
\nonumber\|l_{j,0}\|_{\delta',s,b}&\leq C \sum_{i=1}^{N}\|m_{i,0}\|_{\delta',s,b}\|u_{j0}-u_{j0,x}\|_{\delta',s,b}\|v_{i0}+v_{i0,x}\|_{\delta',s,b}\\
\nonumber&\leq C\sum_{i=1}^{N}\|W_0\|_{\delta',s+2,b}(\|W_0\|_{\delta',s,b}+\|W_0\|_{\delta',s+1,b})^2\\
&\leq C\|W_0\|_{1,s+2,b}^{3},
\end{align}
and 
\begin{align}\label{zj0}
\|z_{j,0}\|_{\delta',s,b}\leq C\|W_0\|_{1,s+2,b}^{3}.
\end{align}
Eqs. (\ref{lj0}) and (\ref{zj0}) indicate 
\begin{align}\label{J3}
\nonumber J_3&=\|(1-\partial_x^2)^{-1}f_{j,0}\|_{\delta',s+2,b}= \|f_{j,0}\|_{\delta',s,b}\\
\nonumber&\leq \|l_{j,0}\|_{\delta',s,b}+\|z_{j,0}\|_{\delta',s,b}\leq C\|W_0\|^3_{1,s+2,b}\\
&\leq \frac{C}{(\delta-\delta')^b}\|W_0\|^3_{1,s+2,b},
\end{align}
while
Eqs. (\ref{J1}), (\ref{J2}) and (\ref{J3}) foster 
\begin{align*}
\|(1-\partial_x^2)^{-1}F_j(W_0,W_{0,x})\|_{\delta',s+2,b}\leq C\frac{\|W_0\|^{k+1}_{1,s+2,b}
+\|W_0\|^3_{1,s+2,b}}{(\delta-\delta')^b}.
\end{align*}
Replacing $\delta'$ by $\delta$ and $\delta$ by 1, and setting
\begin{align*}
M=C(\|W_0\|^{k+1}_{1,s+2,b}+\|W_0\|^3_{1,s+2,b}),
\end{align*}
we have
\begin{align*}
&~~~\|(1-\partial_x^2)^{-1}F_j(W_0,W_{0x})\|_{\delta,s+2,b}\leq \frac{M}
{(1-\delta)^b},~~1\leq j \leq N.
\end{align*}
Adopting a procedure similar to the estimate of $\|(1-\partial_x^2)^{-1}F_j(W_0,W_{0x})\|_{\delta,s+2,b}$ produces
\begin{align*}
&~~~\|(1-\partial_x^2)^{-1}G_j(W_0,W_{0x})\|_{\delta,s+2,b}\leq \frac{M}
{(1-\delta)^b},~~1\leq j \leq N.
\end{align*}
Thus, we attain the desired estimate (\ref{P0}).

Substituting the above $L$ and $M$ into (\ref{TT}), and letting $R=\|W_0\|_{1,s+2,b}$, we obtain
\begin{align*}
\frac{(2^{b}-1)R}{(2^{b}-1)2^{2b+3}LR + M} > \frac{1}{2^{2b+4}L},
\end{align*}
which gives
\begin{align}\label{lifespan}
T'&= \frac{1}{2^{2b+4}L}=\frac{1}{C(2^{2b+4+k}\|W_0\|_{1,s+2,b}^{k}+2^{2b+6}\|W_0\|_{1,s+2,b}^{2})}\\
&\approx\frac{1}{\|W_0\|_{\delta,s+2,b}^k+\|W_0\|_{\delta,s+2,b}^2}.
\end{align}
So, by Theorem 2.1,
there exists a unique solution $W(t)$ to the Cauchy problem (\ref{mtivp1}), which  is
a holomorphic vector function in $(D(0, T'(1 - \delta)^b))^{2N} \mapsto (G^{\delta,s,b})^{2N}$ for every $\delta \in(0, 1)$.
 This completes the proof of Lemma 3.1 and thus the proof of Theorem 3.1 as well.
\end{proof}
\begin{remark3}
From the proof of Theorem 3.1, we see that the lifespan is affected by the order of the nonlinear terms included in (\ref{mt}).
\end{remark3}


\section{Continuity of the data-to-solution map}
\newtheorem{remark}{Remark}[section]
\newtheorem{example}{Example}[section]
\newtheorem{theorem}{Theorem}[section]
\newtheorem{definition}{Definition}[section]
\newtheorem{lemma}{Lemma}[section]

\par
In this section, we prove the continuity of the data-to-solution
map for the initial data and solution described in Theorem 3.1.

First, let us recall that the scale of Banach spaces $X_{\delta}$ and the function
$F(u)$ satisfy the conditions (1) and (2).  For any $\delta > 0$ and $T>0$, let $H(|t| < T;X_{\delta})$ denote the set of holomorphic functions in $|t| < T$ valued in $X_{\delta}$. Also, we know 
that for $ 0<\delta\leq 1$ and $g\in H(|t| < T;X_{\delta})$ with $T > 0$, the following Cauchy problem 
\begin{align}
\frac{df}{dt}=g,~~f(0)=f_0
\end{align}
has a unique solution $f \in H(|t| < T;X_{\delta})$ given by
\begin{align}
f(t)=f_0+\int_0^t g(\tau)d\tau.
\end{align}
Therefore, it follows that the existence of $w$ in Theorem 2.1 is equivalent to the existence of
$w \in H(|t| < \frac{T'(1-\delta)^b}{2^b-1};X_{\delta})$, for every $ \delta\in (0, 1)$, if
$|t| < \frac{T'(1-\delta)^b}{2^b-1}$
\begin{align}
\|\int_0^t Fw(\tau)d\tau\|_{\delta}<R.
\end{align}
So,  our initial value problem (\ref{mtivp1}) is converted to finding the fixed point of the following equation
\begin{align}\label{fix}
\nonumber W(t)&=W_0+\int_0^t(1-\partial_x^2)^{-1} P(W,W_x)(\tau)d\tau.
\end{align}
To see that, we need to introduce a new space $E_a$.
\begin{definition}\cite{Luo}
For any $a>0$, define $E_a=\bigcap _{0<\delta<1}
H(D(0, a\frac{(1 - \delta)^b}{2^b-1});X_{\delta})$ as the Banach space of all functions $t\mapsto u(t)$ with the following property:
\begin{align}
u: \{t:|t|<a\frac{(1 - \delta)^b}{2^b-1}\}\rightarrow X_{\delta}~~~~is~~holomorphic,~~0<\delta<1,
\end{align}
whose norm is given by
\begin{align}\label{Ea}
|||u|||_{a} := \mathop {\sup }\limits_{\scriptstyle |t|<a\frac{(1 - \delta)^b}{2^b-1} \hfill \atop
  \scriptstyle 0 < \delta < 1 \hfill}\{\|u(t)\|_{\delta}(1-\delta)^b\sqrt{1-\frac{|t|}{a(1-\delta)^b}} \}<\infty.
\end{align}
\end{definition}
\begin{remark}
Apparently, we have $E_{T_2}\hookrightarrow E_{T_1}$ for any
$0<T_1<T_2$.
\end{remark}


As per the procedure illustrated in \cite{Luo}, the following lemmas hold in the space $E_a$ with the norm (\ref{Ea}).

\begin{lemma}\label{aaa}
Let $b\geq1$, $a>0$, $u\in E_a$, $0 < \delta < 1$, and $0<t<\frac{a(1-\delta)^b}{2^b-1}$. Then we have
\begin{align*}
\int_0^{t}\frac{\|u(\tau )\|_{\delta(\tau)}}{(\delta(\tau)-\delta)^b} d\tau\leq \frac{a2^{2b+3}|||u|||_a}{(1-\delta)^b}\sqrt{\frac{a(1-\delta)^b}{a(1-\delta)^b-|t|}},
\end{align*}
where $\delta(\tau)\in (\delta,1)$ is given by
\begin{align}\label{del} \delta(\tau)=\frac{1}{2}(1+\delta)+(\frac{1}{2})^{2+\frac{1}{b}}{[(1-\delta)^b-\frac{t}{a}]^{1/b}
-[(1-\delta)^b+(2^{b+1}-1)\frac{t}{a}]^{1/b}}.
\end{align}
\end{lemma}

\begin{lemma}\label{Ku}
Let $b\geq1$, $a>0$, $0 < \delta < 1$, $0<t<\frac{a(1-\delta)^b}{2^b-1}$, and
$W,Y\in \overline{B(W_0,R)}\subset (E_a)^{2N}$. Then, under assumption (\ref{con2}) the following inequality holds:
\begin{align*}
\|\int_0^t(1-\partial_x^2)^{-1} [P(W,W_x)-P(Y,Y_x)](\tau)d\tau\|_{\delta}\leq L
\int_0^{t}\frac{\|W(\tau )-Y(\tau)\|_{\delta(\tau)}}{(\delta(\tau)-\delta)^b} d\tau,
\end{align*}
where $\delta(\tau)$ defined in Lemma \ref{aaa} is a continuous function on $[0, |t|]$ satisfying $\delta<\delta(\tau)\leq 1,$
and $L$ is the same constant as shown in condition (\ref{con2}).
\end{lemma}

For our needs, let us also recall the continuity definition of the data-to-solution map.
\begin{definition}\cite{B-H-P1,Luo}\label{MT}
 The data-to-solution map $W_0\mapsto W(t)$ is called "continuous" if for a given
$W_{\infty}(0) \in (G^{1,s+2,b})^{2N}$ there exist $T > 0$ and $R > 0$ such that for any sequence
of initial data $W_{l}(0) \in (G^{1,s+2,b})^{2N}$ converging to $W_{\infty}(0)$ in $(G^{1,s+2,b})^{2N}$, the corresponding solutions
 $W_{l}(t),W_{\infty}(t)$ to the Cauchy problem (\ref{mtivp1}) for all sufficiently large $l$ satisfy: $W_{l}(t),W_{\infty}(t)\in (E_{T})^{2N}$ and $|||W_{l}(t)-W_{\infty}(t)|||_T\rightarrow 0$, where
\begin{align*}
|||W|||_T:=\sup\{\|W(t)\|_{\delta}(1-\delta)^b\sqrt{1-\frac{|t|}{T(1-\delta)^b}}:0<
\delta<1,|t|<\frac{T(1-\delta)^b}{2^b-1}\}.
\end{align*}
\end{definition}

Let us now give the continuity theorem of the solution map for the Cauchy problem (\ref{mtivp1}).
\begin{theorem}\label{cc}
Let $b\geq1$, $W_0 \in (G^{1,s+2,b})^{2N}, s > \frac{1}{2}$, and $R > 0$. Then
there exists $T > 0$, which is
given in (\ref{TTTT}), such that the Cauchy problem (\ref{mtivp1}) has a unique solution
$
W\in (E_{T})^{2N} ,
$
and
the data-to-solution map $W_0\mapsto W(t): (G^{1,s+2,b})^{2N}\mapsto (E_{T})^{2N}$ is continuous.
\end{theorem}

\begin{proof}
Let $b\geq1$,  $s > \frac{1}{2}$, and $W_{\infty}(0)\in (G^{1,s+2,b})^{2N}$. Assume $W_{l}(0)\in (G^{1,s+2,b})^{2N}$ is a sequence of initial data converging to $W_{\infty}(0)$, namely,  $\|W_{l}(0)-W_{\infty}(0)\|_{1,s,b}\rightarrow 0$, as $l\rightarrow \infty$. Then, there exists an integer $N_0\in\mathbb{N}$, such that for any $l\geq N_0$, we have
\begin{align}\label{un}
\|W_{l}(0)\|_{1,s+2,b}\leq\|W_{\infty}(0)\|_{1,s+2,b}+1.
\end{align}

Let
\begin{align}\label{RR}
R_{\infty}=\|W_{\infty}(0)\|_{1,s+2,b}+1,
 \end{align}
 and
 \begin{align}\label{RRR}
\nonumber R_{l}&=R_{\infty}+\|W_{l}(0)-W_{\infty}(0)\|_{1,s+2,b}\\
&\leq R_{\infty}+1,~~for ~~l\geq N_0.
 \end{align}
 For the given inial data $W_{\infty}(0),W_{l}(0)\in (G^{1,s+2,b})^{2N}$, Theorem 3.1 guarantees the existence and uniqueness of the corresponding solutions $W_{\infty}(t)\in E_{T_{W_{\infty}(0)},R_{\infty}}$ and $W_{l}(t)\in E_{T_{W_{l}(0)},R_{l}}$ respectively. For $|t|<\frac{T_{W_{\infty}(0)}(1-\delta)^b}{2^b-1}$,
\begin{align}\label{m1}
W_{\infty}(t)=W_{\infty}(0)+\int_0^t(1-\partial_x^2)^{-1} P(W_{\infty},W_{\infty,x})(\tau)d\tau,
\end{align}
for $|t|<\frac{T_{W_{l}(0)}(1-\delta)^b}{2^b-1}$,
\begin{align}\label{m2}
W_{l}(t)=W_{l}(0)+\int_0^t(1-\partial_x^2)^{-1} P(W_{l},W_{l,x})(\tau)d\tau.
\end{align}
and their lifespans are given by
  \begin{align*}
&~~~~T_{W_{\infty}(0)} =\frac{1}{2^{2b+4}(C_1\|W_{\infty}(0)\|_{1,s+2,b}^k+C_2\|W_{\infty}(0)\|_{1,s+2,b}^2)},
\end{align*} and
  \begin{align*}
&~~~~T_{W_{l}(0)} =\frac{1}{2^{2b+4}(C_1\|W_{l}(0)\|_{1,s+2,b}^k+C_2\|W_{l}(0)\|_{1,s+2,b}^2)},
\end{align*} respectively,
where $C_1=2^kC$, $C_2=4C$,  and $C$ is a constant coming from (\ref{lifespan}).
The common lifespan $T$ of $W_{\infty}(t)$ and $W_{l}(t)$ can be determined through the following formula
  \begin{align}\label{TTTT}
&~~~~T:=T_{W_{\infty}(0),W_{l}(0)} =\frac{1}{2^{2b+4}(C_1R^k+C_2R^2)},
\end{align}
where
\begin{align}\label{RRRR}
R=2R_{\infty}+2.
\end{align}
Eqs. (\ref{un}), (\ref{RR}), and (\ref{RRRR}) imply that $T_{W_{\infty}(0),W_{l}(0)}< T_{W_{\infty}(0)}$ and $T_{W_{\infty}(0),W_{l}(0)}< T_{W_{l}(0)}$, that is, $T_{W_{\infty}(0),W_{l}(0)}< \min\{ T_{W_{\infty}(0)}, T_{W_{l}(0)}\}. $

As $l\geq N_0$, noticing $E_{T_{W_{\infty}(0), R_{\infty}}}\hookrightarrow E_{T_{W_{\infty}(0),W_{l}(0)},R}$ and $E_{T_{W_{l}(0), R_l}}\hookrightarrow E_{T_{W_{\infty}(0),W_{l}(0)},R}$, we have
$W_{\infty}(t),W_{l}(t)\in E_{T_{W_{\infty}(0),W_{l}(0)},R}$. Moreover, for  $0<\delta<1$, $b\geq1$, $|t|<\frac{T_{W_{\infty}(0),W_{l}(0)}(1-\delta)^b}{2^b-1}$ and $l\geq N_0$, one has
\begin{align*}
\nonumber&~~~~\|W_{l}(t ) - W_{\infty}(0)\|_{\delta,s+2,b}\\
\nonumber &\leq \|W_{l}(t ) - W_{l}(0)\|_{\delta,s+2,b}+\|W_{l}(0) - W_{\infty}(0)\|_{1,s+2,b}\\
&\leq R_{l}+1\leq R_{\infty}+2<R.
\end{align*}

Next, we aim to prove $|||W_l - W_{\infty}|||_{T_{W_{\infty}(0),W_{l}(0)}}\rightarrow 0$ as $l\rightarrow \infty$.
When $0<\delta<1$ and $|t|<\frac{T_{W_{\infty}(0),W_{l}(0)}(1-\delta)^b}{2^b-1}$, it follows from Eqs. (\ref{m1}) and (\ref{m2}) that
\begin{align}
\nonumber&~~~~||W_l - W_{\infty}||_{\delta,s+2,b}\\
\nonumber&\leq \int_0^{t}
\|(1-\partial_x^2)^{-1}[P(W_l,W_{l,x})-P(W_{\infty},W_{\infty,x})]\|_{\delta,s+2,b}d\tau\\
&~~~~~~+\|W_l(0) - W_{\infty}(0)\|_{\delta,s+2,b} .
\end{align}

For $0<\delta<1$, $|t|<\frac{T_{W_{\infty}(0),W_{l}(0)}(1-\delta)^b}{2^b-1}$, $0\leq |\tau|=\tau \leq |t|$, and $\delta<\delta(\tau)\leq 1$,
we need to prove that
\begin{align*}
\|W_{\infty}(\tau ) - W_{l}(0)\|_{\delta(\tau),s+2,b}<R,
\end{align*}
and
\begin{align*}
\|W_{l}(\tau ) - W_{l}(0)\|_{\delta(\tau),s+2,b}<R.
\end{align*}
As $l\geq N_0$, we have
\begin{align}\label{mn2}
\nonumber&~~~~\|W_{\infty}(\tau ) - W_{l}(0)\|_{\delta(\tau),s+2,b}\\
\nonumber &\leq \|W_{\infty}(\tau ) - W_{\infty}(0)\|_{\delta(\tau),s+2,b}+\|W_{\infty}(0) - W_{l}(0)\|_{\delta(\tau),s+2,b}\\
&<R_{\infty}+1<R.
\end{align}
Theorem 2.1 and Eq. (4.9) guarantee the following result
\begin{align}\label{mn22}
\|W_{l}(\tau ) - W_{l}(0)\|_{\delta(\tau),s+2,b}\leq R_{l}<R.
\end{align}
Eqs. (\ref{mn2}), (\ref{mn22}), (\ref{con1}), and Lemma 4.2 lead to
\begin{align}\label{mm1}
\nonumber&\|W_l(t) - W_{\infty}(t)\|_{\delta,s+2,b}-\|W_l(0) - W_{\infty}(0)\|_{\delta,s+2,b}\\
\nonumber&\leq\int_0^{t}
\|(1-\partial_x^2)^{-1}[P(W_l,W_{l,x})-P(W_{\infty},W_{\infty,x})]\|_{\delta,s+2,b}d\tau\\
&\leq
L_l
\int_0^{t}\frac{\|W_{\infty}(\tau )-W_{l}(\tau ))\|_{\delta(\tau),s,b}}{(\delta(\tau)-\delta)^b} d\tau,
\end{align}
where $L_l=C_1(R_l+\|W_0\|_{1,s,b})^k
+C_2(R_l+\|W_0\|_{1,s,b})^2$ and $\delta(\tau)$ is defined in (\ref{del}), and  $a={T_{W_{\infty}(0),W_{l}(0)}})$.

In the light of  Lemma \ref{aaa} with $a=T_{W_{\infty}(0),W_{l}(0)}$ and (\ref{mm1}), under the conditions $0<\delta<1$ and $|t|<T_{W_{\infty}(0),W_{l}(0)}\frac{(1-\delta)^b}{2^b-1}$, we have
\begin{align*}
&\|W_l(t) - W_{\infty}(t)\|_{\delta,s,b}-\|W_l(0) - W_{\infty}(0)\|_{\delta,s,b}\\
&\leq
\frac{2^{2b+3}T_{W_{\infty}(0),W_{l}(0)}L_l|||W_{\infty}-W_l|||_{T_{W_{\infty}(0),W_{l}(0)}}}
{(1-\delta)^b}\sqrt{\frac{T_{W_{\infty}(0),W_{l}(0)}(1-\delta)^b}{T_{W_{\infty}(0),W_{l}(0)}(1-\delta)^b-|t|}},
\end{align*}
which implies 
\begin{align*}
\nonumber &~~~|||W_l(t) - W_{\infty}(t)|||_{T_{W_{\infty}(0),W_{l}(0)}}\\
&\leq2^{2b+3}T_{W_{\infty}(0),W_{l}(0)}L_l|||W_{\infty}-W_l|||_{T_{W_{\infty}(0),W_{l}(0)}}+\|W_l(0) - W_{\infty}(0)\|_{\delta,s,b}.
\end{align*}
Therefore, we have 
\begin{align}
\nonumber&~~~(1-2^{2b+3}T_{W_{\infty}(0),W_{l}(0)}L_l)|||W_l(t) - W_{\infty}(t)|||_{T_{W_{\infty}(0),W_{l}(0)}}\\
&\leq\|W_l(0) - W_{\infty}(0)\|_{\delta,s,b}.
\end{align}
So, as $l\geq N_0$, Eqs. (\ref{RR}), (\ref{RRR}) and (\ref{RRRR}) generate 
\begin{align}\label{ln}
\nonumber L_l&=C_1(R_l+\|W_l(0)\|_{1,s,b})^k
+C_2(R_l+\|W_l(0)\|_{1,s,b})^2\\
\nonumber&\leq C_1(R_l+\|W_{l}(0)-W_{\infty}(0)\|_{1,s,b}+\|W_{\infty}(0)\|_{1,s,b})^k\\
\nonumber&~~~~+C_2(R_l+\|W_{l}(0)-W_{\infty}(0)\|_{1,s,b}+\|W_{\infty}(0)\|_{1,s,b})^2\\
\nonumber&\leq C_1(R_l+1+\|W_{\infty}(0)\|_{1,s,b})^k
+C_2(R_l+1+\|W_{\infty}(0)\|_{1,s,b})^2\\
\nonumber&= C_1(R_l+R_{\infty})^k
+C_2(R_l+R_{\infty})^2\\
&\leq C_1R^k
+C_2R^2.
\end{align}
Due to $T_{W_{\infty}(0),W_{l}(0)}=\frac{1}{2^{2b+4}(C_1R^k+C_2R^2)}$ and Eq. (\ref{ln}), we have
\begin{align*}
2^{2b+3}T_{W_{\infty}(0),W_{l}(0)}L_l
<\frac{1}{2},
\end{align*} which implies 
 \begin{align}
 \nonumber &~~~|||W_l(t) - W_{\infty}(t)|||_{T_{W_{\infty}(0),W_{l}(0)}}\\
&\leq2\|W_l(0) - W_{\infty}(0)\|_{1,s,b}.
  \end{align}
Therefore we complete the proof of Theorem 4.1.
\end{proof}

\begin{remark}
Theorem \ref{cc} includes important results since it makes the 2n-component
Camassa-Holm system (\ref{mtivp1}) well-posed
in the spaces $G^{\delta,s,b}$
in the sense of Hadamard. One may compare these results with the
classical Cauchy-Kovalevski theorem, where there is no show-up of continuity for the data-to-solution map.
\end{remark}

Applying those results in Theorem \ref{local} and Theorem \ref{cc}, it is not difficult for one to obtain the well-posedness properties for
the following two integrable systems with arbitrary polynomial functions.

{\bf 1. A four-component Camassa-Holm system}

We may propose the Cauchy problem for the following four-component Camassa-Holm type system with initial data $W_0=(u_{1,0},u_{2,0},v_{1,0},v_{2,0})^T\in (G^{1,s,b})^4$:
\begin{equation}\left\{\begin{array}{ll}\label{4CHivp}
m_{1,t}+(Hm_1)_x+n_2(g_1g_2-H)+m_1(f_2g_2+2f_1g_1)=0,\\
m_{2,t}+(Hm_2)_x-n_1(g_1g_2-H)-m_2(f_1g_1+2f_2g_2)=0,\\
n_{1,t}+(Hn_1)_x-m_2(f_1f_2-H)-n_1(f_2g_2+2f_1g_1)=0,\\
n_{2,t}+(Hn_2)_x+m_1(f_1f_2-H)+n_2(f_1g_1+2f_2g_2)=0,\\
m_{j}(0,x)=m_{j,0}(x),\\
n_{j}(0,x)=n_{j,0}(x),~~ j=1,2,\end{array}\right.
\end{equation}
where $m_{j}=u_{j}-u_{j,xx},
n_{j}=v_{j}-v_{j,xx},$  $f_1=u_2-v_{1,x}, \ f_2=u_1+v_{2,x}, \ g_1=v_2+u_{1,x}$, and $g_2=v_1-u_{2,x}$.
The system (\ref{4CHivp}) is a special case of the multi-component CH model proposed by Xia and Qiao \cite{XQ1} in 2013, and was also studied
in 2014 by Li, Liu and Popowicz \cite{LLP}.
This system has Lax pair, 
bi-Hamiltonian structure, and infinitely many conservation laws \cite{LLP,XQ1}. Recently, Zhang and Yin \cite{ZY1} studied the local well-posedness for the system in Besove spaces, also they presented several global existence and blow-up results for two integrable two-component subsystems by selecting specific polynomials.

The free choice of the arbitrary function $H$ in (\ref{4CHivp}) allows one to recover
some well-known CH type equations through reductions with the detailed procedure in \cite{XQ1}.
The four-component system (\ref{4CHivp}) is completely integrable for any function $H$.
To study the well-posedness problem in the Gevrey class spaces, as discussed in Section 3, let us take the arbitrary  
$H$ as a degree $k$ polynomial function
 of $u_1,u_2,v_1,v_2$ and their derivatives. Similar to Eq. (\ref{mt}),
 the four-component system (\ref{4CHivp}) has the degree $(k+1)$ terms and cubic terms.
 Of course, 
 (\ref{4CHivp}) is naturally  a special case of (\ref{mt}) with $N=2$ \cite{XQ1}.
 With a slight modification in the proofs of Theorem \ref{local} and Theorem \ref{cc},
 it is not difficult for us to obtain the following local well-posedness
  for the Cauchy problem of the four-component system (\ref{4CHivp}).

\begin{theorem}
Let $s > \frac{1}{2}$ and $b\geq1 .$ If $W_0 \in (G^{1,{s+2},b})^{4}$,
then there exists a positive
time $T$, which depends on the initial data $W_0$, such that for every $\delta \in (0, 1),$ the Cauchy
problem (\ref{4CHivp}) has a unique solution $W(t)=(u_1(t),u_2(t),v_1(t),v_2(t))^T$ which is a holomorphic vector function in the disc $D(0, \frac{T(1-\delta)^b}{2^b-1})$
valued
in $(G^{\delta,{s+2},b})^{4}$.
 Furthermore,
the lifespan
 $$T\approx \frac{1}{\|W_0\|_{\delta,s+2,b}^k+\|W_0\|_{\delta,s+2,b}^2},$$ and the data-to-solution map $W_0\mapsto W(t): (G^{1,s,b})^{4}\mapsto (E_{T})^{4}$ is continuous.
\end{theorem}


{\bf 2. A two-component Camassa-Holm system}

Another amazing Camassa-Holm type system with an arbitrary function is the following
 synthetical integrable two-component peakon model 
 proposed by Xia, Qiao and Zhou \cite{XQZ}.
\begin{equation}\left\{\begin{array}{ll}\label{2CHivp}
m_{t}=F+F_x-\frac{1}{2}(uv-u_xv_x+uv_x-u_xv),\\
n_{t}=-G+G_x+\frac{1}{2}(uv-u_xv_x+uv_x-u_xv),\\
\end{array}\right.
\end{equation}
where $m=u-u_{xx}$ and
$n=v-v_{xx}.$
As studied in \cite{XQZ}, the above two-component model was proved integrable through its Lax pair and infinitely
many conservation laws. Moreover, the authors investigated the bi-Hamiltonian structure and the
interaction of multi-peakons.
Apparently, if $F = mH, G = nH,$
where $H$ is an arbitrary polynomial in $u, v$ and their derivatives, (\ref{2CHivp}) is reduced to the
special case of (\ref{mtivp}) with $N=1$. Thus the local well-posedness for the Cauchy problem of
(\ref{2CHivp}) naturally holds in Theorem \ref{local} and Theorem \ref{cc}.
The two-component model (\ref{2CHivp}) has very interesting solutions including muti-peakon solutions, muti-weak-kink solutions, weak kink-peakon solutions,
and periodic-peakon solutions. More interestingly,  (\ref{2CHivp})
is the first integrable system with peakon solutions, within our knowledge, which are not in the traveling wave type.
Therefore, we will give a deep exploration in analysis for this remarkable system (\ref{2CHivp}) in next section.

\section{Unique continuation }
\newtheorem{theorem1}{Theorem}[section]
\newtheorem{lemma1}{Lemma}[section]
\newtheorem {remark1}{Remark}[section]
\newtheorem{corollary1}{Corollary}[section]
\newtheorem{definition1}{Definition}[section]
\par
In this section, we focus on the following Cauchy problem of thesystem (\ref{mt}) with $N=  1$, that is system (\ref{2CHivp}) with $F = mH$ and $G = nH,$
 \begin{equation}\label{special1}
\left\{\begin{array}{ll}m_{t}=(mH)_x+mH+\frac{1}{2}[m(u-u_x)(v+v_x)] ,&t > 0,\,x\in\mathbb{R}, \\
n_{t}=(nH)_x-nH-\frac{1}{2}[n(u-u_x)(v+v_x)], &t>0, x\in\mathbb{R},\\
m(0,x)=m_0,~n(0,x)=n_0, &x\in\mathbb{R} ,
\end{array}\right.
\end{equation}
where $m=u-u_{xx}$ and $n=v-v_{xx}$.
 The property of unique continuation for the Cauchy problem (\ref{special1}) may be reflected with compactly supported initial data.
In the case of compactly supported initial data, unique continuation is
essentially an infinite speed of propagation of its support. Therefore, one may naturally ask the question: how will a strong solution behave at infinity under the given compactly
supported initial data? To provide a sufficient answer, let us prepare two lemmas listed below.

Given the initial data $z_0=(u_0,v_0)^T \in H^{s}\times
H^{s}, s\geqslant 3, $ Theorem 3.1 in \cite{ZY} ensures the local well-posedness
of strong solutions $z=(u,v)^T$ for Eq. (\ref{special1}). Consider the following initial value problem
\begin{equation}\label{qH}
\left\{\begin{array}{ll}q_{t}(t,x) = -H(z,z_x)(t,q(t,x)) ,&t \in [0,T),\,x\in \mathbb{R}, \\
q(0,x) = x, &x\in \mathbb{R},\end{array}\right.
\end{equation}
where $u,v$ denote the two components of solution $z$ to Eq.(\ref{special1}).
Due to $z(t,.)\in H^3\times H^3 \subset C^m$ with $0\leq m \leq \frac{5}{2},$
$z=(u,v)^T \in C^{1}([0,T)\times \mathbb{R},\mathbb{R}).$ Applying the
classical results in the theory of ordinary differential equations,
one can obtain the following result 
which is the key in the
proof for unique continuation of strong solutions to Eq.(3.1).

Let us first present two lemmas as follows.
\begin{lemma1}\label{lemma1}
 Let $z_{0}\in H^{s}\times H^{s}$, $s \geq 3$. Then Eq.(\ref{qH}) has a
unique solution $q \in C^{1}([0,T)\times\mathbb{R},\mathbb{R}).$
Moreover, for $(t,x)\in[0,T)\times\mathbb{R}$, the map $q(t,\cdot)$ is an increasing diffeomorphism over
$\mathbb{R}$ with
\begin{align}\label{qx}
q_{x}(t,x)=\exp\left(\int_{0}^{t}-H_x(s,q(s,x))ds\right)>0.
\end{align}
\begin{proof}
Differentiating Eq. (\ref{qH}) with respect to $x$ yields
\begin{equation}\label{qx1}
\left\{\begin{array}{ll}
\frac{d}{dt}q_x=-H_xq_{x},&t>0,~~x\in \mathbb{R},\\
q_x(0,x) = 1, &x\in \mathbb{R}.\end{array}\right.
\end{equation}
Solving the above equation, we obtain
$$
q_x(t,x)=\exp\{\int_0^t-H_x(s,q(s,x))ds\}>0,
$$
which exactly reads as Eq. (\ref{qx}).
\end{proof}

\end{lemma1}
\par
\begin{lemma1}\label{lemma2}
Let $z_{0}\in H^{s}\times H^{s}$, $s \geq 3$, and $T>0$ be the
maximal existence time of corresponding solution $z$ to Eq.(\ref{special1}).
Then for all $(t,x)\in
[0,T)\times\mathbb{R}$ we have
\begin{align}
\label{m} m(t,q(t,x))q_{x}(t,x) = m_{0}(x)\exp\{\int_0^t[H+\frac{1}{2}(u-u_x)(v+v_x)](s,q(s,x))ds\}, \\
\label{n} n(t,q(t,x))q_{x}(t,x) = n_{0}(x)\exp\{\int_0^t-[H+\frac{1}{2}(u-u_x)(v+v_x)](s,q(s,x))ds\}.
\end{align}
\begin{proof}
By Eqs. (\ref{special1}), (\ref{qH}) and  (\ref{qx1}), a direct calculation reveals 
\begin{align*}
&~~\frac{d}{dt}[m(t,q(t,x))q_x]
=(m_t + m_x q_t)q_x+m q_{tx}\\
&=[m_t - m_x H-mH_x)]q_x\\
&= (m_t -(mH)_x)q_x\\
&=[mH+\frac{1}{2}m(u-u_x)(v+v_x)]q_x\\
&=[H+\frac{1}{2}(u-u_x)(v+v_x)]mq_x.
\end{align*}
Solving for $mq_x$ from the above equation, we obtain
$$
m(t,q(t,x))q_x(t,x)=m_0(x)\exp\{\int_0^t[H+\frac{1}{2}(u-u_x)(v+v_x)](s,q(s,x))ds\},
$$
which is exactly Eq. (\ref{m}).

Adopting a similar procedure to the second equation in (\ref{special1})shown as above, we have
\begin{align*}
&~~\frac{d}{dt}[n(t,q(t,x))q_x]
=(n_t + n_x q_t)q_x+n q_{tx}\\
&=[n_t - n_x H-nH_x)]q_x\\
&= (n_t -(nH)_x)q_x\\
&=-[nH+\frac{1}{2}n(u-u_x)(v+v_x)]q_x\\
&=-[H+\frac{1}{2}(u-u_x)(v+v_x)]nq_x,
\end{align*}
which leads to  Eq. (\ref{n}).
\end{proof}
\end{lemma1}

Let us  now determine the behavior of solutions at infinity 
through the following theorem.

\begin{theorem1}\label{good}
Let $z=(u,v)^T\in C([0,T),H^s)\times C([0,T),H^s),s> \frac{5}{2}$ be a nontrivial solution of (\ref{special1}) with the maximal existence time
$T > 0,$ which has a compactly supported initial data on the interval $[a, b].$ Then we have
\begin{equation}
u(t,x)=\left\{\begin{array}{ll}\frac{1}{2}E_+(t)e^{-x} ,&x>q(t,b), \\
\frac{1}{2}E_-(t)e^{x} ,&x<q(t,a),\\
\end{array}\right.
\end{equation}
\begin{equation}
v(t,x)=\left\{\begin{array}{ll}\frac{1}{2}F_+(t)e^{-x} ,&x>q(t,b), \\
\frac{1}{2}F_-(t)e^{x} ,&x<q(t,a),\\
\end{array}\right.
\end{equation}
where $E_+(t):=\int_{q(t,a)}^{q(t,b)}e^y m(t,y)dy$, $E_-(t):=\int_{q(t,a)}^{q(t,b)}e^{-y} m(t,y)dy$, $F_+(t):=\int_{q(t,a)}^{q(t,b)}e^y n(t,y)dy$ and $F_-(t):=\int_{q(t,a)}^{q(t,b)}e^{-y} n(t,y)dy.$
Moreover, $E_+(t)$, $E_-(t)$, $F_+(t)$ and $F_-(t)$
are continuously non-vanishing with $E_+(0) = E_-(0)=F_+(0)=F_-(0)= 0$, and 
the monotonicity of $E_+$ and $F_-$ is displayed in the following four cases:

(1) If $m_0 \geq 0$ and $n_0\geq 0$, then $E_+$ is strictly
increasing, while $F_-$ is strictly decreasing for $t\in [0, T).$

(2) If $m_0\leq 0$ and $n_0\leq 0$, then $E_+$ is strictly
decreasing,
 while $F_-$ is strictly increasing for $t\in [0, T).$

(3) If $m_0 \geq 0$ and $n_0\leq 0$, then both $E_+$ and $F_-$ are strictly
decreasing for $t\in [0, T).$

(4) If $m_0 \leq 0$ and $n_0\geq 0$, then both $E_+$ and $F_-$ are strictly
increasing for $t\in [0, T).$
\end{theorem1}
\begin{remark1}
Theorem \ref{good} covers our previous work \cite{hu}, where
 the unique continuation of the system with $N=1$ and $H=-\frac{1}{2}(uv-u_xv_x)$ in (\ref{special1}) was proved.
\end{remark1}

For each fixed $t>0$, Theorem \ref{good} also tells us that as long as the solution $z=(u,v)^T$ exists, then the sign of the first component $u(t,x)$ at positive infinity is determined by the sign of $m_0$, while the sign of the second component $v(t, x)$ at negative infinity is determined by the sign of $n_0$.
 Let us now give the proof of Theorem 5.1.

\begin{proof}
If $u_0$ and $v_0$ are initially supported on the compact interval $[a,b]$,  so are $m_0$ and $n_0$.
From Eqs. (\ref{m}) and (\ref{n}), it follows
that $m( t,\cdot),n(t,\cdot)$ are compact with their support belonging to the interval $[q(t,a), q(t,b)].$
One may readily use the relation $u = \frac{1}{2} e^{-|x|} \ast m$ and $v = \frac{1}{2} e^{-|x|} \ast n$ to write
\begin{align}
u(t,x)=\frac{e^{-x}}{2}\int_{-\infty}^x e^y m(t,y)dy +\frac{e^x}{2}\int_{x}^{\infty} e^{-y} m(t,y)dy,\\
u_x(t,x)=-\frac{e^{-x}}{2}\int_{-\infty}^x e^y m(t,y)dy +\frac{e^x}{2}\int_{x}^{\infty} e^{-y} m(t,y)dy,
\end{align}
and
\begin{align}
v(t,x)=\frac{e^{-x}}{2}\int_{-\infty}^x e^y n(t,y)dy +\frac{e^x}{2}\int_{x}^{\infty} e^{-y} n(t,y)dy,\\
v_x(t,x)=-\frac{e^{-x}}{2}\int_{-\infty}^x e^y n(t,y)dy +\frac{e^x}{2}\int_{x}^{\infty} e^{-y} n(t,y)dy.
\end{align}
Assume that $m_0$ and $n_0$ are non-negative, then we have
\begin{align*}
u(t,x)+u_x(t,x)=\frac{e^x}{2}\int^{\infty}_x e^y m(t,y)dy\geq 0, \\
u(t,x)-u_x(t,x)=\frac{e^{-x}}{2}\int_{-\infty}^x e^y m(t,y)dy\geq 0,\\
v(t,x)+v_x(t,x)=\frac{e^x}{2}\int^{\infty}_x e^y n(t,y)dy\geq 0, \\
v(t,x)-v_x(t,x)=\frac{e^{-x}}{2}\int_{-\infty}^x e^y n(t,y)dy\geq 0.
\end{align*}
Let us define the following four functions
\begin{align*}
E_+(t)=\int_{q(t,a)}^{q(t,b)}e^y m(t,y)dy,~~E_-(t)=\int_{q(t,a)}^{q(t,b)}e^{-y} m(t,y)dy,\\
F_+(t)=\int_{q(t,a)}^{q(t,b)}e^y n(t,y)dy,~~F_-(t)=\int_{q(t,a)}^{q(t,b)}e^{-y} n(t,y)dy.
\end{align*}
Apparently, we have 
\begin{align}\label{EF}
\nonumber u(t,x)=\frac{e^{-x}}{2}E_+(t),~~x>q(t,b), \\
\nonumber u(t,x)=\frac{e^{x}}{2}E_-(t),~~x<q(t,a),\\
\nonumber v(t,x)=\frac{e^{-x}}{2}F_+(t),~~x>q(t,b), \\
v(t,x)=\frac{e^{x}}{2}F_-(t),~~x<q(t,a).
\end{align}
Therefore, differentiating both sides of Eq. (\ref{EF}) leads to
\begin{align}
\nonumber \frac{e^{-x}}{2}E_+(t)=u(t,x)=-u_x(t,x)=u_{xx}(t,x),~~x>q(t,b), \\
\nonumber \frac{e^{x}}{2}E_-(t)=u(t,x)=u_x(t,x)=u_{xx}(t,x),~~x<q(t,a),\\
\nonumber \frac{e^{-x}}{2}F_+(t)=v(t,x)=-v_x(t,x)=v_{xx}(t,x),~~x>q(t,b), \\
\frac{e^{x}}{2}F_-(t)=v(t,x)=v_x(t,x)=v_{xx}(t,x),~~x<q(t,a).
\end{align}
Since $u(0,\cdot)$ and $v(0,\cdot)$ are supported in the interval $[a, b]$, this immediately gives us $E_+(0) = E_-(0) = 0$ and $F_+(0) = F_-(0) = 0$.

Due to $m(t,\cdot)$ supported in the interval $[q(t,a), q(t, b)]$, then for each fixed $t$, we have
\begin{align*}
 &\frac{dE_+(t)}{dt}=\int_{q(t,a)}^{q(t,b)}e^y m_t(t,y)dy\\
 &=\int_{-\infty}^{\infty}e^y m_t(t,y)dy\\
  &=\int_{-\infty}^{\infty}[(mH)_y+mH+\frac{1}{2}(u-u_y)(v+v_y)m]e^ydy\\
 &=\int_{-\infty}^{\infty}\frac{1}{2}(u-u_y)(v+v_y)me^ydy>0.
\end{align*}
Furthermore,
\begin{align*}
 &\frac{dF_-(t)}{dt}=\int_{q(t,a)}^{q(t,b)}e^{-y} n_t(t,y)dy\\
 &=\int_{-\infty}^{\infty}e^{-y} n_t(t,y)dy\\
 &=\int_{-\infty}^{\infty}[(nH)_y-nH-\frac{1}{2}(u-u_y)(v+v_y)n]e^{-y}dy\\
&=-\int_{-\infty}^{\infty}\frac{1}{2}(u-u_y)(v+v_y)ne^{-y}ydy<0,
\end{align*}
where the strict monotonicity described above follows from our assumption that the solution is
nontrivial. 
With a completely similar procedure to the above proofs, we can also achieve the monotonicity results of $E_+$ and $F_-$ in other
three cases.
Thus, we complete the proof of Theorem 5.1.
\end{proof}

By the fine structure of the two-component system (\ref{special1}), we may determin the sign of the solution  $u$ at positive infinity, and the sign of $v$ at negative infinity through investigating
the monotonicity of $E_+$ and $F_-$, respectively.
However, the properties of the solution of $u$ at negative infinity, and $v$ at positive infinity can not be determined due to the unknown monotonicity of $E_-$ and $F_+$. The main obstacle is the existence of the
arbitrary function $H$ in (\ref{special1}). But, we may choose some special polynomial function $H$ so as to
guarantee the monotonicity of $E_+$, $F_-$, $F_+$, and $E_-$.

 Letting $H=-\frac{1}{2}(u-u_x)(v+v_x)$ sends Eq. (\ref{special1}) to the SQQ equation (\ref{2CH2}).
 We now propose the Cauchy problem for the SQQ equation as follows:
 \begin{equation}\label{special2}
\left\{\begin{array}{ll}m_{t}+\frac{1}{2}[m(u-u_x)(v+v_x)]_x=0 ,&t > 0,\,x\in\mathbb{R}, \\
n_{t}+\frac{1}{2}[n(u-u_x)(v+v_x)]_x=0, &t>0, x\in\mathbb{R},\\
m(0,x)=m_0,~n(0,x)=n_0, &x\in\mathbb{R}.
\end{array}\right.
\end{equation}
 Substituting $H=-\frac{1}{2}(u-u_x)(v+v_x)$ into (\ref{qH}) yields 
 \begin{equation}\label{q}
\left\{\begin{array}{ll}q_{t}(t,x) = \frac{1}{2}(u-u_x)(v+v_x)(t,q(t,x)) ,&t \in [0,T),\,x\in \mathbb{R}, \\
q(0,x) = x, &x\in \mathbb{R},\end{array}\right.
\end{equation}
Therefore, 
Lemma \ref{lemma1} and Lemma \ref{lemma2} recover Lemma \ref{lemma11} and Lemma \ref{lemma22}, respectively, which were studied in \cite{ZY}.
 \begin{lemma1}\label{lemma11}
 Let $z_{0}\in H^{s}\times H^{s}$ and $s \geq 3$. Then Eq.(\ref{q}) has a
unique solution $q \in C^{1}([0,T)\times\mathbb{R},\mathbb{R}).$
Moreover, in $(t,x)\in[0,T)\times\mathbb{R}$ the map $q(t,\cdot)$ is an increasing diffeomorphism over
$\mathbb{R}$ with
\begin{align}
q_{x}(t,x)=\exp\left(\int_{0}^{t}\frac{1}{2}[m(v+v_x)-n(u-u_x)](s,q(s,x))ds\right)>0.
\end{align}
\end{lemma1}

\begin{lemma1}\label{lemma22}
Let $z_{0}\in H^{s}\times H^{s}$, $s \geq 3$, and $T>0$ be the
maximal existence time of the solution $z$ corresponding to Eq.(\ref{special2}).
Then for all $(t,x)\in
[0,T)\times\mathbb{R}$, we have
\begin{eqnarray}
 m(t,q(t,x))q_{x}(t,x) = m_{0}(x), \\
 n(t,q(t,x))q_{x}(t,x) = n_{0}(x).
\end{eqnarray}
 \end{lemma1}

Let us now give the unique continuation for the SQQ system (\ref{special2}).
\begin{theorem1}\label{good1}
Let $z=(u,v)^T\in C([0,T),H^s)\times C([0,T),H^s),s> \frac{5}{2},$ be a nontrivial solution of (\ref{special2}) with the maximal existence time
$T > 0,$ which has an initial data compactly supported on the interval $[a, b].$ Then, we have
\begin{equation}
u(t,x)=\left\{\begin{array}{ll}\frac{1}{2}E_+(t)e^{-x} ,&x>q(t,b), \\
\frac{1}{2}E_-(t)e^{x} ,&x<q(t,a),\\
\end{array}\right.
\end{equation}
\begin{equation}
v(t,x)=\left\{\begin{array}{ll}\frac{1}{2}F_+(t)e^{-x} ,&x>q(t,b), \\
\frac{1}{2}F_-(t)e^{x} ,&x<q(t,a),\\
\end{array}\right.
\end{equation}
where $E_+(t):=\int_{q(t,a)}^{q(t,b)}e^y m(t,y)dy$ and $E_-(t):=\int_{q(t,a)}^{q(t,b)}e^{-y} m(t,y)dy$, $F_+(t):=\int_{q(t,a)}^{q(t,b)}e^y n(t,y)dy$ and $F_-(t):=\int_{q(t,a)}^{q(t,b)}e^{-y} n(t,y)dy.$
Moreover, $E_+(t)$, $E_-(t)$, $F_+(t)$ and $F_-(t)$
are continuously non-vanishing  with $E_+(0) = E_-(0)=F_+(0)=F_-(0)= 0$, and 
the monotonicity of $E_+$, $E_-$, $F_+$ and $F_-$ is displayed in the following four cases:

(1) If $m_0 \geq 0$ and $n_0\geq 0$, then $E_+$ and $F_+$ are strictly
increasing, while $E_-$ and $F_-$ are strictly decreasing for $t\in [0, T).$

(2) If $m_0\leq 0$ and $n_0\leq 0$, then $E_+$ and $F_+$ are strictly
decreasing,
 while $E_-$ and $F_-$ are strictly increasing for $t\in [0, T).$

(3) If $m_0 \geq 0$ and $n_0\leq 0$, then both $E_+$ and $F_-$ are strictly
decreasing, while $E_-$ and $F_+$ are strictly increasing for $t\in [0, T).$

(4) If $m_0 \leq 0$ and $n_0\geq 0$, then both $E_+$ and $F_-$ are strictly
increasing, while $E_-$ and $F_+$ are strictly decreasing for $t\in [0, T).$
\end{theorem1}
\begin{proof}
As discussed above, it is sufficient for us to prove the monotonicity of $F_+$ and $E_-$ only. Since $n(t,\cdot)$ is compactly supported in the interval $[q(t,a), q(t, b)]$, for each fixed $t$ we have
\begin{align}
 \frac{dF_+(t)}{dt}=\int_{q(t,a)}^{q(t,b)}e^y n_t(t,y)dy=\int_{-\infty}^{\infty}e^y n_t(t,y)dy,
\end{align}
which implies 
\begin{align*}
 &\frac{dF_+(t)}{dt}=\int_{q(t,a)}^{q(t,b)}e^y n_t(t,y)dy\\
 &=\int_{-\infty}^{\infty}e^y n_t(t,y)dy\\
 &=-\int_{-\infty}^{\infty}\frac{1}{2}[(u-u_y)(v+v_y)n]_ye^ydy\\
 &=\int_{-\infty}^{\infty}\frac{1}{2}(u-u_y)(v+v_y)ne^ydy>0.
\end{align*}
A similar way yields
\begin{align*}
 &\frac{dE_-(t)}{dt}=\int_{q(t,a)}^{q(t,b)}e^{-y} m_t(t,y)dy\\
 &=\int_{-\infty}^{\infty}e^{-y} m_t(t,y)dy\\
 &=-\int_{-\infty}^{\infty}\frac{1}{2}[((u-u_y)(v+v_y)m]_ye^{-y}dy\\
 &=-\int_{-\infty}^{\infty}\frac{1}{2}(u-u_y)(v+v_y)me^{-y}ydy< 0.
\end{align*}
The other three cases can be proved analogously.
This completes the proof of Theorem \ref{good1}.
\end{proof}

\bigskip
\noindent\textbf{Acknowledgments}


This work  was partially supported by the National Natural Science Foundation of China (Nos. 11401223, 11171295 and 61328103), NSF of Guangdong (No. 2015A030313424) and China Scholarship Council. The first author would like to thank Dr. Zhijun Qiao for his kind hospitality and encouragement during her visit in the University of Texas - Rio Grande Valley, and also thank so much to Dr. Qiao's team members for their fruitful discussion.  The second author also thanks the Haitian Scholar Plan of Dalian University of Technology, the China state administration of foreign experts affairs system under the affiliation of China University of Mining and Technology, and the U.S. Department of Education GAANN project (P200A120256) for their cooperation in conducting the research program.

\end{document}